\documentclass[]{article}  
\usepackage{url,float}
\usepackage{graphicx}
\usepackage{amsmath}
\usepackage{amsfonts}
\usepackage{amssymb}
\usepackage{latexsym}
\usepackage{hyperref}
\usepackage{enumerate}

\let\realbfseries=\bfseries
\def\bfseries{\realbfseries\boldmath}

\newcommand{\hide}[1]{}

\newcommand{\ABox}{
\raisebox{3pt}{\framebox[6pt]{\rule{6pt}{0pt}}}
}
\newenvironment{proof}{{\bf Proof:}}{\hfill\ABox}

\newtheorem{theorem}{{\bf Theorem}}

\newtheorem{lemma}{Lemma}
\newtheorem{question}{Question}

\newcommand{\lemlab}[1]{\label{lemma:#1}}
\newcommand{\thmlab}[1]{\label{thm:#1}}
\newcommand{\figlab}[1]{\label{fig:#1}}
\newcommand{\seclab}[1]{\label{sec:#1}}

\newcommand{\lemref}[1]{\ref{lemma:#1}}

\newcommand{\figref}[1]{\ref{fig:#1}}

\def\Z{{\mathbb{Z}}}

\usepackage{xcolor}

\def\emph#1{\textit{\textbf{\boldmath #1}}}

\title{%
Some Polycubes Have No Edge Zipper Unfolding\footnote{
An extended abstract  of this paper is to appear in the
\emph{Canad. Conf. Comput. Geom.}, Aug. 2020.}
} 

\author{%
Erik D. Demaine%
     \thanks{MIT Computer Science and Artificial Intelligence Laboratory,
       Cambridge, MA 02139, USA,
       \protect\url{{edemaine,mdemaine}@mit.edu}}
\and
   Martin L. Demaine\footnotemark[1]
\and
David Eppstein%
\thanks{
Computer Science Department
University of California, Irvine, CA 92679, USA
\protect\url{eppstein@uci.edu}
Supported in part by NSF grants  CCF-1618301 and CCF-1616248}
\and
Joseph O'Rourke%
    \thanks{Department of Computer Science, Smith College, Northampton, MA
      01063, USA.
      \protect\url{jorourke@smith.edu}.}
}

\begin{document}
\maketitle

\begin{abstract}
It is unknown whether every polycube (polyhedron constructed by gluing cubes
face-to-face) has an edge unfolding, that is, cuts along edges of the cubes
that unfolds the polycube to a single nonoverlapping polygon in the plane.
Here we construct 
polycubes that have no \emph{edge zipper unfolding}
where the cut edges are further restricted to form a path.
\end{abstract}

\section{Introduction}
\seclab{Introduction}

A \emph{polycube} $P$ is an object constructed by gluing cubes whole-face to whole-face,
such that its surface is a manifold. 
Thus the neighborhood of every surface point is a disk; so there are no edge-edge nor
vertex-vertex nonmanifold surface touchings.
Here we only consider polycubes of genus zero.
The \emph{edges} of a polycube are all the cube edges on the surface, even when those edges
are shared between two coplanar faces.
Similarly, the \emph{vertices} of a polycube are all the cube vertices on the surface, even when those vertices
are \emph{flat}, incident to $360^\circ$ total face angle.
Such polycube flat vertices have degree $4$.
It will be useful to distinguish these flat vertices from 
\emph{corner vertices}, nonflat vertices with total incident angle
$\neq 360^\circ$ (degree $3$, $5$, or $6$).
For a polycube $P$, let its \emph{$1$-skeleton graph $G_P$} include every
vertex and edge of $P$, with vertices marked as either corner or flat.

It is an open problem to determine whether every polycube has an
\emph{edge unfolding} (also called a \emph{grid unfolding}) ---
a tree in the $1$-skeleton that spans all corner vertices (but need not include flat vertices)
which, when cut, unfolds the surface to a \emph{net}, a planar nonoverlapping polygon~\cite{JORsurvey}.
By \emph{nonoverlapping} we mean that no two points, each interior
to a face, are mapped to the same point in the plane.
This definition allows two boundary edges to coincide in the net,
so the polygon may be ``weakly simple.''
The intent is that we want to be able to cut out the net and refold to $P$.

It would be remarkable if every polycube could be edge unfolded, but no
counterexample is known.
There has been considerable exploration of orthogonal polyhedra, a more general type
of object, for which there are examples that cannot be edge-unfolded~\cite{bddloorw-uscop-98}.
(See~\cite{damian2018unfolding} for citations to earlier work.)
But polycubes have more edges in their $1$-skeleton graphs for the cut tree
to follow than do orthogonal polyhedra,
so it is conceivably easier to edge-unfold polycubes.

A restriction of edge unfolding studied
in~\cite{s-cpcn-75, lddss-zupc-10, o-fzupp-10, ddu-zudp-13}
is \emph{edge zipper unfolding} (also called \emph{Hamiltonian unfolding}).
A \emph{zipper} unfolding has a cut tree that is a path
(so that the surface could be ``unzipped'' by a single zipper).
It is apparently unknown whether even the highly restricted edge zipper
unfolding could unfold every polycube to a net.
The result of this note is to settle this question in the negative:
polycubes are constructed none of which have an
edge zipper unfolding.
Two polycubes in particular,
shown in Fig.~\figref{TwoPolycubes},
have no such unfolding. Other polycubes with the same property
are built upon these two.
\begin{figure}[htbp]
\centering
\includegraphics[width=0.75\linewidth]{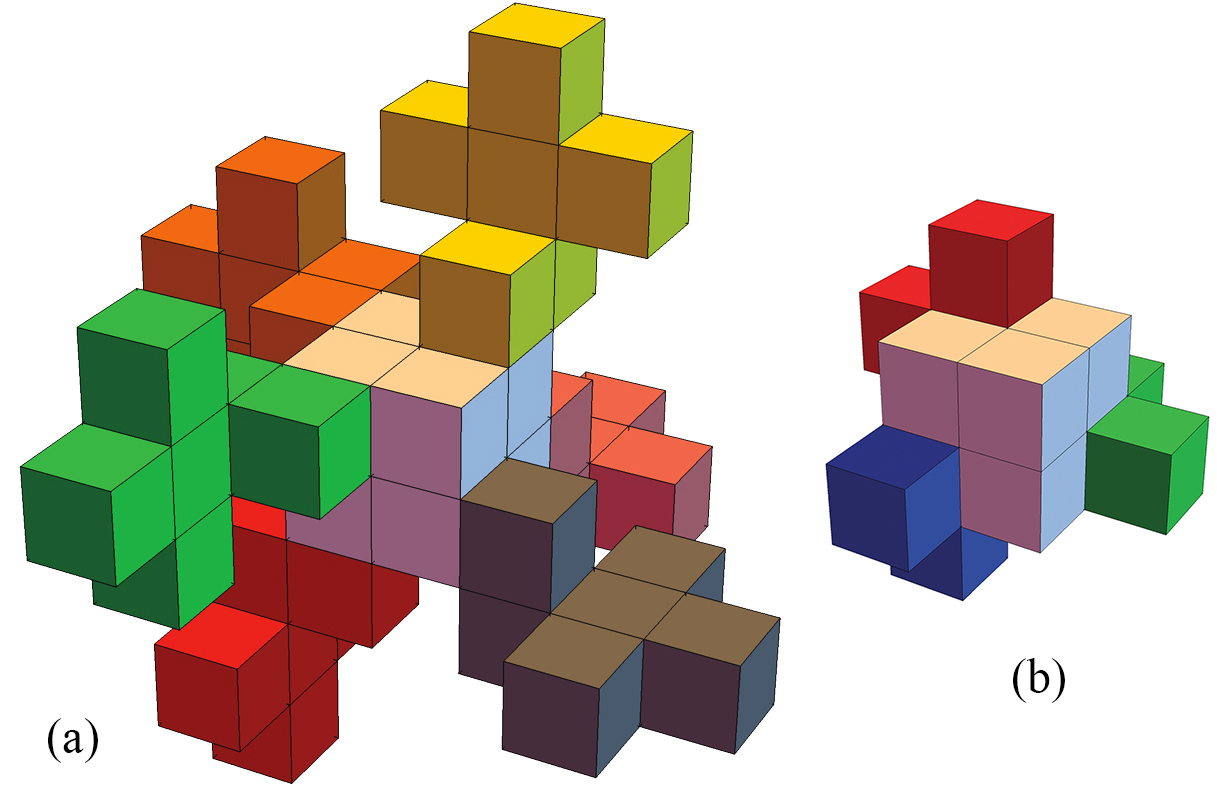}
\caption{Two polycubes that have no edge zipper unfolding.}
\figlab{TwoPolycubes}
\end{figure}

\section{Hamiltonian Paths}
\seclab{HamiltonianPaths}
Shephard~\cite{s-cpcn-75} introduced Hamiltonian unfoldings of convex polyhedra,
what we refer to here as edge zipper unfolding, following the terminology
of~\cite{lddss-zupc-10}.
Any edge zipper unfolding must cut along a Hamiltonian path of the vertices.
It is easy to see that not every convex polyhedron has an edge zipper unfolding,
simply because the rhombic dodecahedron has no Hamiltonian path.
This counterexample avoids confronting the difficult nonoverlapping condition.

We follow a similar strategy here, constructing a polycube with no Hamiltonian path.
But there is a difference in that a polycube edge zipper unfolding need not
include flat vertices, and so need not be a Hamiltonian path in $G_P$.
Thus identifying a polycube $P$ that has no Hamiltonian path does not
immediately establish that $P$ has no edge zipper unfolding,
if $P$ has flat vertices.

So one approach is to construct a polycube  $P$ that has no flat vertices---every
vertex is a corner vertex. Then, if $P$ has no Hamiltonian path, then it has
no edge zipper unfolding.
A natural candidate is the polycube object $P_6$ shown in
Fig.~\figref{P6}.
\begin{figure}[htbp]
\centering
\includegraphics[width=0.4\linewidth]{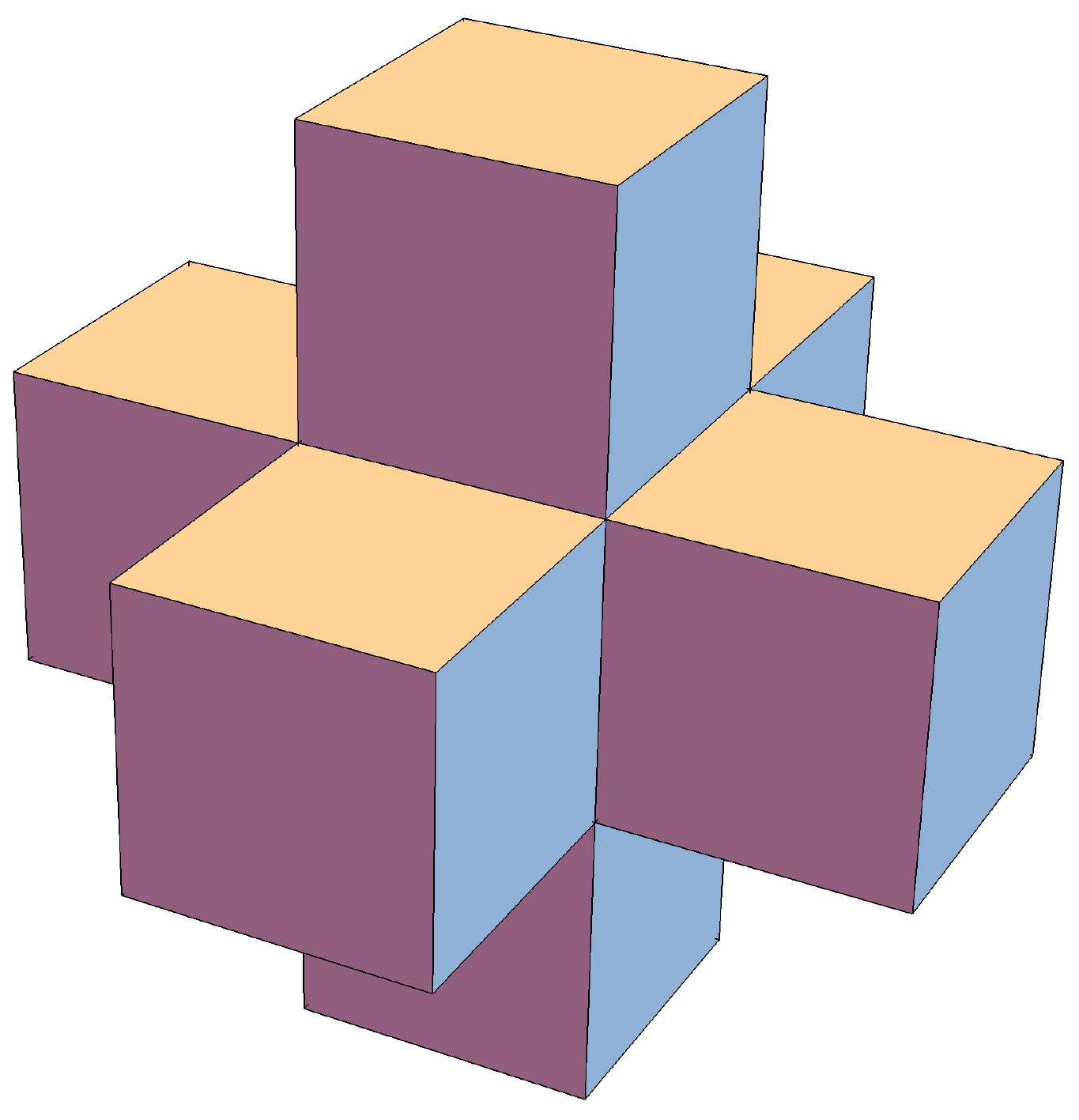}
\caption{All of $P_6$'s vertices are corner vertices.}
\figlab{P6}
\end{figure}
However, the $1$-skeleton of $P_6$ does admit Hamiltonian paths, and indeed
we found a path that unfolds $P_6$ to a net.

Let $\overline{G}_P$ be the dual graph of $P$: each cube is a node, and two nodes are
connected if they are glued face-to-face.
A \emph{polycube tree} is a polycube whose dual graph is a tree. $P_6$ is a polycube tree.
That it has a Hamiltonian path is an instance of a more general claim:

\begin{lemma}
The graph $G_P$ for any polycube tree $P$ has a Hamiltonian cycle.
\lemlab{trees}
\end{lemma}
\begin{proof}
It is easy to see by induction that every polycube tree can be built by gluing cubes
each of which touches just one face at the time of gluing: never is there a
need to glue a cube to more than one face of the previously built object.

A single cube has a Hamiltonian cycle. Now assume that every polycube tree 
of $\le n$ cubes has a Hamiltonian cycle. For a tree $P$ of $n+1$ cubes,
remove a $\overline{G}_P$ leaf-node cube $C$, and apply the induction hypothesis.
The exposed square face $f$ to which $C$ glues to make $P$ includes either $2$ or $3$ edges of
the Hamiltonian cycle ($4$ would close the cycle; $1$ or $0$ would imply the cycle
misses some vertices of $f$).
It is then easy to extend the Hamiltonian cycle to include $C$, as shown in 
Fig.~\figref{HamLem}.
\end{proof}
\begin{figure}[htbp]
\centering
\includegraphics[width=0.65\linewidth]{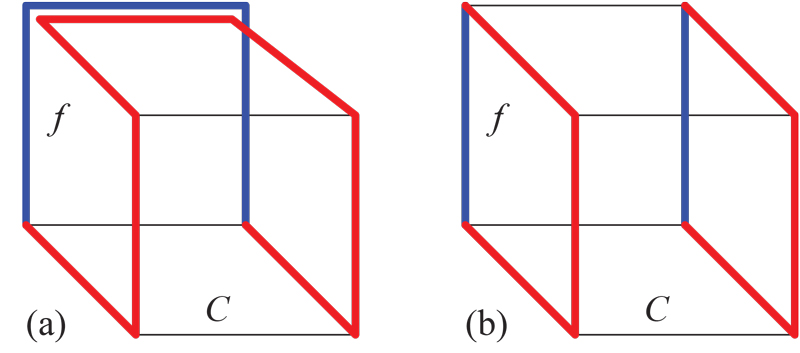}
\caption{(a)~$f$ contains $3$ edges of the cycle (blue);
(b) $f$ contains $2$ edges of the cycle. The cycles are extended to $C$ by 
replacing the blue with the the red paths.}
\figlab{HamLem}
\end{figure}

\medskip
\noindent
So to prove that a polycube tree has no edge zipper unfolding
would require an argument that confronted nonoverlap.
This leads to an open question:
\begin{question}
Does every polycube tree have an edge zipper unfolding?
\end{question}

\section{Bipartite $G_P$}
\seclab{Bipartite}
To guarantee the nonexistence of Hamiltonian paths, we can exploit the bipartiteness
of $G_P$, using Lemma~\lemref{parity-imbalance} below.

\begin{lemma}
A polycube graph $G_P$ is $2$-colorable, and therefore bipartite.
\lemlab{2-colorable}
\end{lemma}
\begin{proof}
Label each lattice point $p$ of $\Z^3$ with the $\{0,1\}$-parity of the
sum of the Cartesian coordinates of $p$.
A polycube $P$'s vertices are all lattice points of $\Z^3$.
This provides a $2$-coloring of $G_P$; $2$-colorable graphs are bipartite.
\end{proof}

The \emph{parity imbalance} in a $2$-colored (bipartite) graph is the
absolute value of the difference in the number of nodes of each color.

\begin{lemma}
A bipartite graph $G$ with a parity imbalance $> 1$ has no
Hamiltonian path.\footnote{
Stated at \url{http://mathworld.wolfram.com/}
\url{HamiltonianPath.html}.
}
\lemlab{parity-imbalance}
\end{lemma}
\begin{proof}
The nodes in a Hamiltonian path alternate colors $010101 \ldots$.
Because by definition a Hamiltonian path includes every node,
the parity imbalance in a bipartite graph with a Hamiltonian path is either $0$ (if of even length)
or $1$ (if of odd length).
\end{proof}

So if we can construct a polycube $P$ that (a)~has no flat vertices,
and (b)~has parity imbalance $>1$, then we will have established
that $P$ has no Hamiltonian path, and therefore no edge zipper unfolding.
We now show that the polycube $P_{44}$, illustrated in
Fig.~\figref{Polycube_44}, meets these conditions.
\begin{figure}[htbp]
\centering
\includegraphics[width=0.75\linewidth]{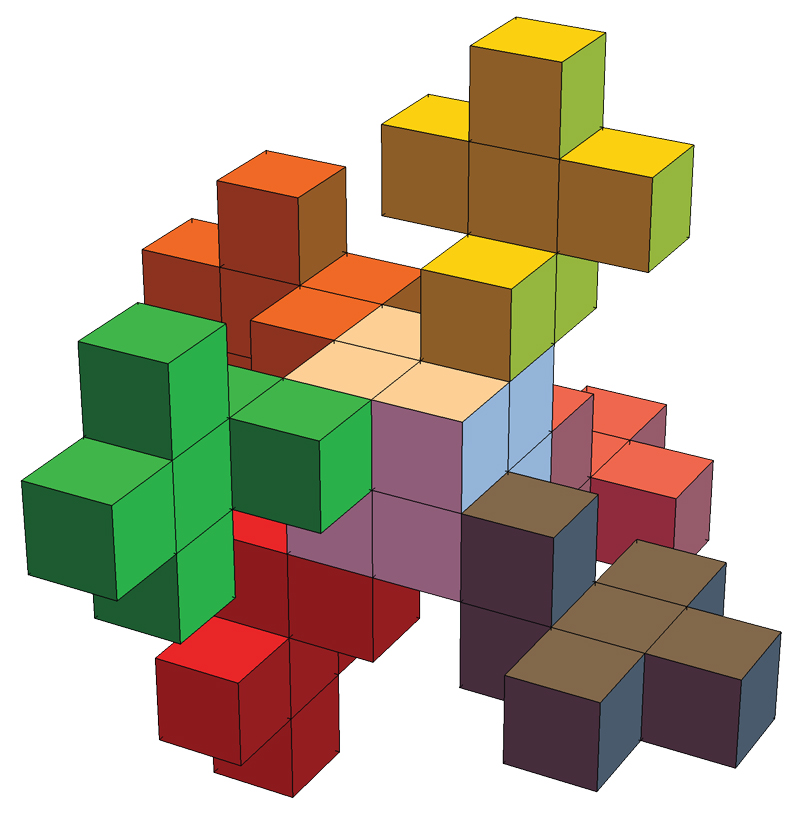}
\caption{The polycube $P_{44}$, consisting of $44$ cubes, has no Hamiltonian path.}
\figlab{Polycube_44}
\end{figure}

\begin{lemma}
The polycube $P_{44}$'s graph $G_{P_{44}}$ has parity imbalance of $2$.
\lemlab{P44parity}
\end{lemma}
\begin{proof}
Consider first the $2 \times 2 \times 2$ cube that is the core of $P_{44}$; call it $P_{222}$. 
The front face $F$ has an extra $0$;
see Fig.~\figref{Parity_2x2}.
\begin{figure}[htbp]
\centering
\includegraphics[width=0.4\linewidth]{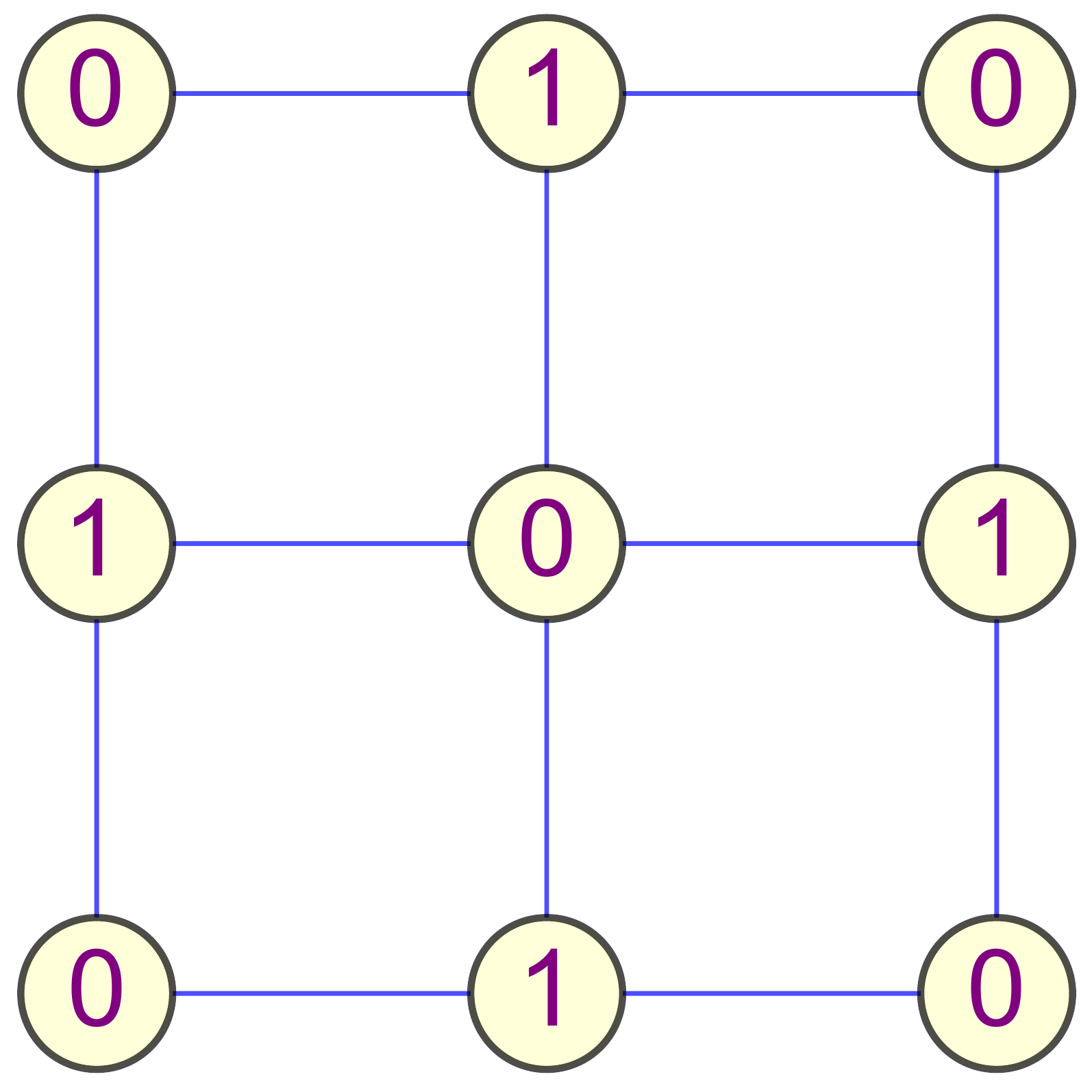}
\caption{$2$-coloring of one face of $P_{222}$.}
\figlab{Parity_2x2}
\end{figure}
It is clear that the $8$ corners of $P_{222}$ are all colored $0$.
The midpoint vertices of the $12$ edges of $P_{222}$ are colored $1$.
Finally the $6$ face midpoints are colored $0$. So $14$ vertices are colored $0$ and $12$ colored $1$.

Next observe that attaching a cube $C$ to exactly one face of any polycube does not change the parity:
the receiving face $f$ has colors $0101$, and the opposite face of $C$ has
colors $1010$.

Now, $P_{44}$ can be constructed by attaching six copies of a
$6$-cube ``cross,'' call it $P_+$, which in isolation is a polycube tree and so can be built
by attaching cubes each to exactly one face. And each $P_+$ attaches to one
corner cube of $P_{222}$.
Therefore $P_{44}$
retains $P_{222}$'s imbalance of $2$.
\end{proof}

\medskip
\noindent
The point of the $P_+$ attachments is to remove the flat vertices of $P_{222}$.
Note that when attached to $P_{222}$, each $P_+$ has only corner vertices.

\begin{theorem}
Polycube $P_{44}$ has no edge zipper unfolding.
\thmlab{NoUnzipping44}
\end{theorem}
\begin{proof}
Although it takes some scrutiny of Fig.~\figref{Polycube_44} to verify, 
$P_{44}$ has no (degree-$4$) flat vertices.
Thus an edge zipper unfolding must pass through every vertex, and so be a Hamiltonian path.
Lemma~\lemref{P44parity} says that $G_{P_{44}}$ has imbalance $2$, and
Lemma~\lemref{parity-imbalance} says it therefore cannot have a Hamiltonian path.
\end{proof}

\section{Construction of $P_{14}$}
\seclab{Construction}
It turns out that the smaller polycube  $P_{14}$ shown in
Fig.~\figref{Polycube_14} also has no edge zipper unfolding, even though it has flat vertices.
\begin{figure}[htbp]
\centering
\includegraphics[width=0.6\linewidth]{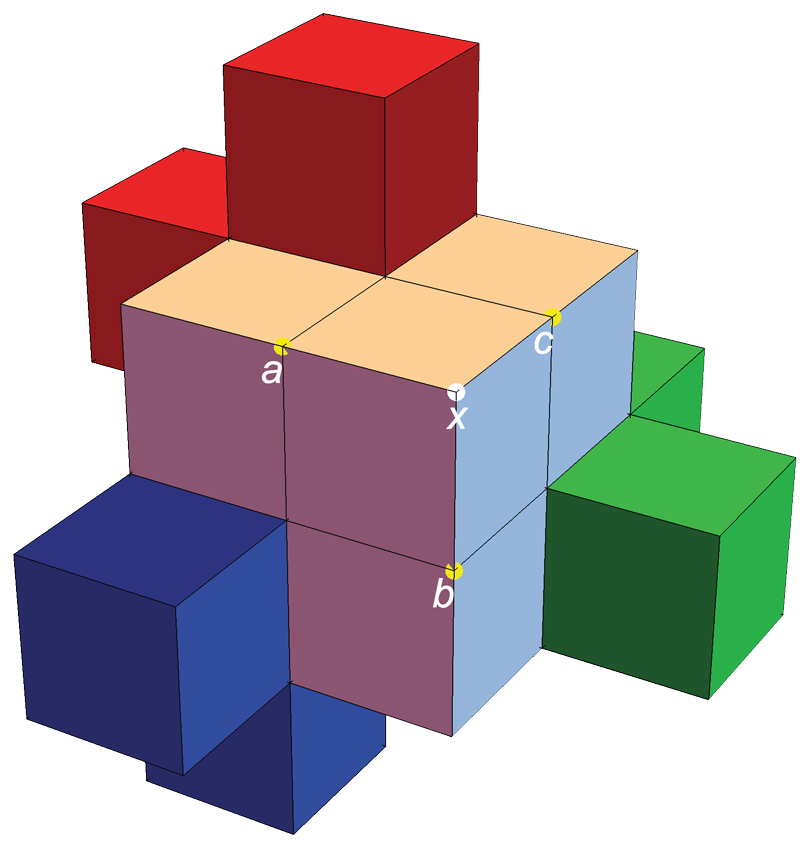}
\caption{$P_{14}$: $P_{222}$ with six $1$-cube attachments.}
\figlab{Polycube_14}
\end{figure}
To establish this, we still need an imbalance $>1$, which easily 
follows just as in 
Lemma~\lemref{P44parity}:
\begin{lemma}
The polycube $P_{14}$'s graph $G_{P_{14}}$ has parity imbalance of $2$.
\lemlab{P*parity}
\end{lemma}

\noindent
But notice that $P_{14}$ has three flat vertices: $a$, $b$, and $c$.

\begin{theorem}
Polycube $P_{14}$ has no edge zipper unfolding.
\thmlab{NoUnzipping14}
\end{theorem}
\begin{proof}
An edge zipper unfolding need not pass through the three flat vertices, $a$, $b$, and $c$, but it could pass through
one, two, or all three.
We show that in all cases, an appropriately modified subgraph of $G_{P_{14}}$ has no Hamiltonian path.
Let $\rho$ be a hypothetical edge zipper unfolding cut path. We consider four exhaustive possibilities,
and show that each leads to a contradiction.
\begin{enumerate}[(1)]
\setcounter{enumi}{-1}
\item \textbf{$\rho$ includes $a,b,c$.}
So $\rho$ is a Hamiltonian path in $G_{P_{14}}$.
But Lemma~\lemref{P*parity} says that $G_{P_{14}}$ has imbalance $2$, and
Lemma~\lemref{parity-imbalance} says that no such graph has a Hamiltonian path.
\item \textbf{$\rho$ excludes one flat vertex $a$ and includes $b,c$.}
(Because of the symmetry of $P_{14}$, it is no loss of generality to assume that it is $a$ that is excluded.)
If $\rho$ excludes $a$, then it does not travel over any of the four edges incident to $a$.
Thus we can delete $a$ from $G_{P_{14}}$; say that $G_{-a}= G_{P_{14}} \setminus a$.
This graph is shown in Fig.~\figref{UnzipHP2x2x2}.
\begin{figure}[htbp]
\centering
\includegraphics[width=0.9\linewidth]{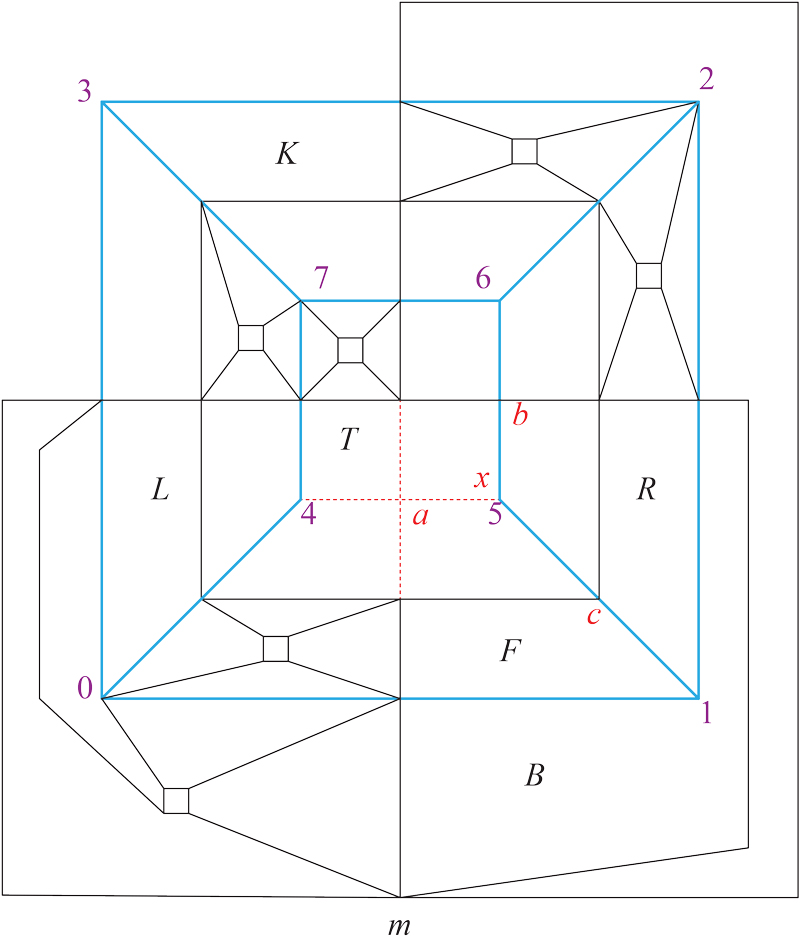}
\caption{Schlegel diagram of $G_{-a}$. 
We follow~\protect\cite{damian2018unfolding} in labeling the
faces of a cube as $F, K, R, L, T, B$ for Front, bacK, Right, Left, Top, Bottom respectively.
The corners of $P_{222}$
are labeled $0,1,2,3$ around the bottom face $B$, and $4,5,6,7$ around the top face $T$.
$m$ is the vertex in the middle of $B$. The edges deleted by removing vertex $a$ are shown
dashed.}
\figlab{UnzipHP2x2x2}
\end{figure}
Following the coloring in Fig.~\figref{Parity_2x2}, all corners of $P_{222}$ are colored $0$,
so each of the edge midpoints $a,b,c$ is colored $1$.
The parity imbalance of $P_{14}$ is $2$ extra $0$'s.
Deleting $a$ maintains bipartiteness and increases the parity imbalance of  $G_{-a}$ to $3$.
Therefore by Lemma~\lemref{parity-imbalance}, $G_{-a}$ has no Hamiltonian path,
and such a $\rho$ cannot exist.
\item \textbf{$\rho$ includes just one flat vertex $c$, and excludes $a,b$.}
(Again symmetry ensures there is no loss of generality in assuming the one included flat vertex is $c$.)
$\rho$ must include corner $x$, which is only accessible in $G_{P_{14}}$ through 
the three flat vertices. If $\rho$ excludes $a,b$, then it must include the edge $cx$.
Let $G_{-ab} = G_{P_{14}} \setminus \{a,b\}$.
In $G_{-ab}$, $x$ has degree $1$, so $\rho$ terminates there.
It must be that $\rho$ is a Hamiltonian path in $G_{-ab}$, but the deletion
of $a,b$ increases the parity imbalance to $4$, and so again such a 
Hamiltonian path cannot exist.
\item \textbf{$\rho$ excludes $a,b,c$.}
Because corner $x$ is only accessible through one of these flat vertices,
$\rho$ never reaches $x$ and so cannot be an edge zipper unfolding.
\end{enumerate}
Thus the assumption that there is an edge zipper unfolding cut path $\rho$
for $P_{14}$ reaches a contradiction in all four cases.
Therefore, there is no edge zipper unfolding cut path for $P_{14}$.\footnote{
Just to verify this conclusion, we constructed these graphs in Mathematica
and \texttt{FindHamiltonianPath[]}  returned \texttt{\{\}} for each.}
\end{proof}

\section{\boldmath Edge Unfoldings of  $P_{14}$ and $P_{44}$}
Now that it is known that $P_{14}$ and $P_{44}$ each have no edge zipper unfolding, it is natural to wonder
whether either settles the edge-unfolding open problem:
can they be edge unfolded? Indeed both can:
see Figures~\figref{P14Unf} 
\begin{figure}[htbp]
\includegraphics[width=0.7\linewidth]{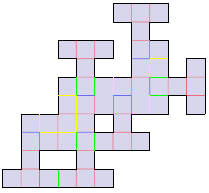}
\caption{Edge unfolding of $P_{14}$.
Colors: green${}={}$cut, red${}={}$mountain, blue${}={}$valley, yellow${}={}$flat.}
\figlab{P14Unf}
\end{figure}
and~\figref{P44Unf}.
The colors in these layouts are those used by Origami Simulator~\cite{ghassaei2018fast}.
\begin{figure}[htbp]
\centering
\includegraphics[width=1.0\linewidth]{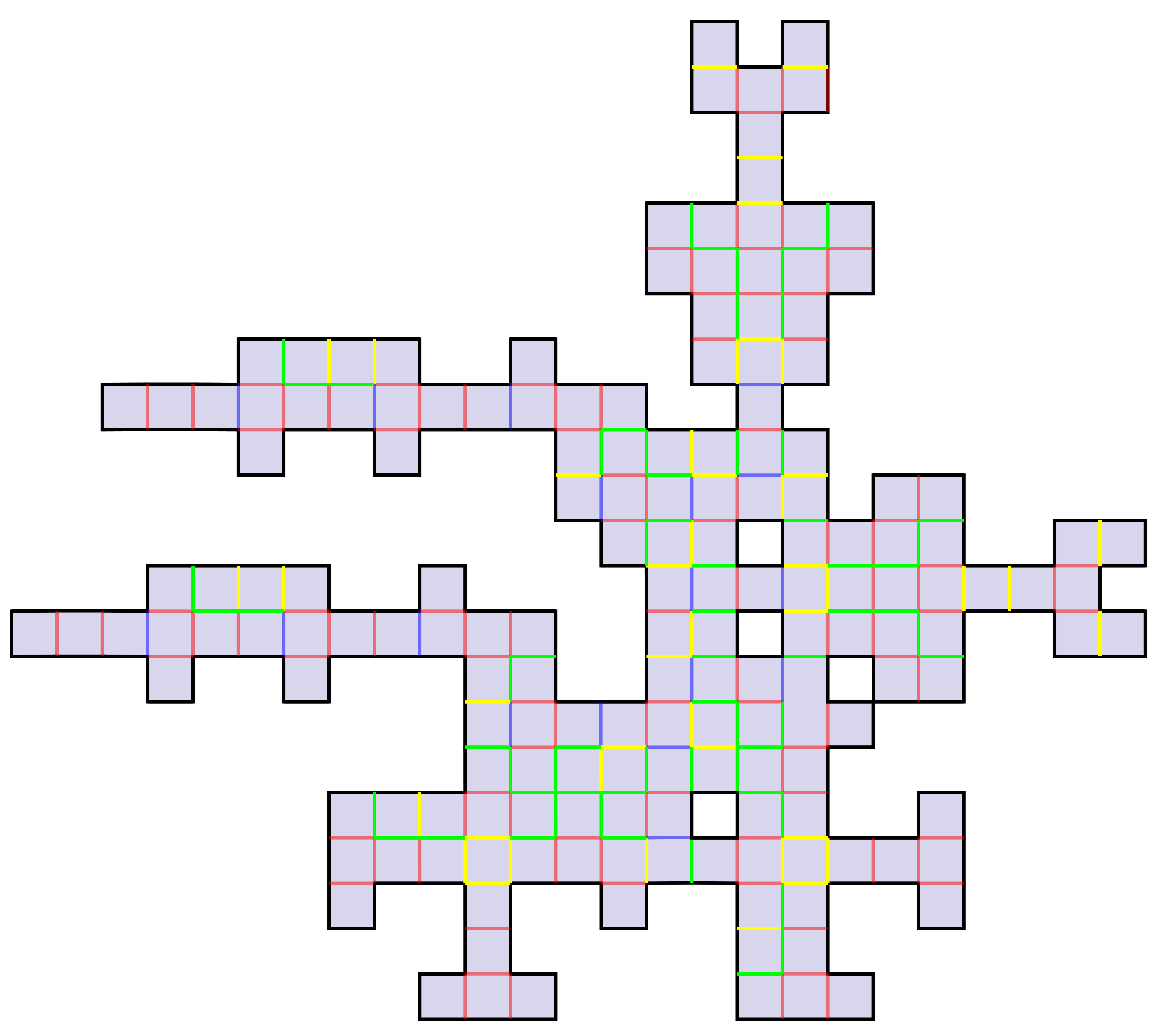}
\caption{Edge unfolding of $P_{44}$.
Colors: green${}={}$cut, red${}={}$mountain, blue${}={}$valley, yellow${}={}$flat.}
\figlab{P44Unf}
\end{figure}
Fig.~\figref{P44_80percent} shows a partial folding of $P_{44}$,
and animations are at \url{http://cs.smith.edu/~jorourke/Unf/NoEdgeUnzip.html}.
\begin{figure}[htbp]
\centering
\includegraphics[width=0.75\linewidth]{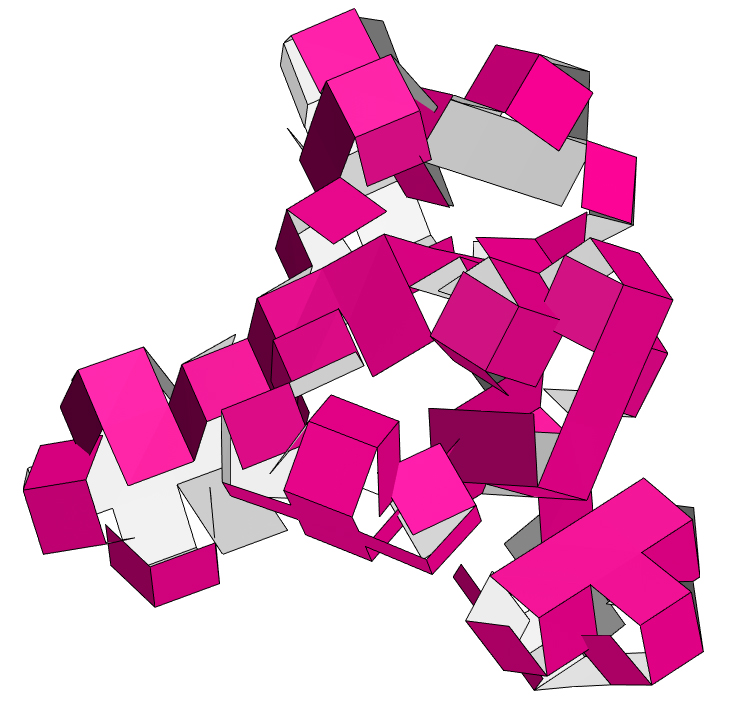}
\caption{Partial folding of the layout in Fig.~\protect\figref{P44Unf}.
Compare with Fig.~\protect\figref{Polycube_44}.}
\figlab{P44_80percent}
\end{figure}

\section{Many Polycubes with No Edge Zipper Unfolding}
As pointed out by Ryuhei Uehara,\footnote{
Personal communication, June 2020.}
$P_{44}$ can be extended to an infinite number of polycubes with no edge zipper
unfolding.
Let $P'_6$ be the polycube in
Fig.~\figref{P6} with the bottom cube removed. So $P'_6$ has a `$+$' sign
of five cubes in its base layer. Let $B$ be the bottom face of the cube at the center
of the `$+$' sign.
Attach $P'_6$ to the highest cube of $P_{44}$ in
Fig.~\figref{TwoPolycubes}(a)
by gluing $B$ to the top face of that top cube. It is easy to verify that
all new vertices of this augmented object, call it $P'_{44}$, are corners.
The joining process can be repeated with another copy of $P'_6$,
producing $P''_{44}$, and so on. All of these polycubes have 
no zipper unfolding.

We have not attempted to edge-unfold these larger objects.

\section{Open Problems}
The most interesting question remaining in this line of investigation
is \textbf{Question~1} (Sec.~2):
Does every polycube tree have an edge zipper unfolding?

\paragraph{Acknowledgements.}
We thank participants of the Bellairs 2018 workshop for their insights.
We benefitted from suggestions by the referees.

\bibliographystyle{alpha}
\bibliography{NoEdgeUnzip}
\end{document}